\documentclass[10pt, conference]{IEEEtran}
\usepackage{amsmath}
\usepackage{amsfonts}
\usepackage[margin=1in]{geometry}
\usepackage{graphicx}
\usepackage{cite}
\usepackage{subcaption}
\usepackage{makecell}
\usepackage{amsthm}
\usepackage{algorithm}
\usepackage[noend]{algpseudocode}

\newtheorem{theorem}{Theorem}

\usepackage{mathtools}

\DeclarePairedDelimiter\ket{\lvert}{\rangle}
\DeclarePairedDelimiterX\braket[2]{\langle}{\rangle}{#1 \delimsize\vert #2}

\renewcommand{\thesubsection}{\thesection.\alph{subsection}}

\usepackage{tikz}
\definecolor{umassMaroon}{rgb}{0.588, 0.004, 0.090} 
\definecolor{edgeColor}{rgb}{0, 0, 1}
\definecolor{localCopyColor}{rgb}{0, 0.5, 0}
\tikzset{
	nodeStyle/.style={draw=red, line width=1pt}
}
\tikzset{
	edgeStyle/.style={draw=edgeColor, line width=1pt}
}
\tikzset{
	localCopyStyle/.style={draw=localCopyColor, line width=1pt}
}
\tikzset{every picture/.style={inner sep=1.5pt}} 

\title{Distributing Graph States Across Quantum Networks}

\author{
    \IEEEauthorblockN{Alex Fischer}\IEEEauthorblockA{College of Information and Computer Sciences\\University of Massachusetts, Amherst\\{\tt alexander.fischer3@gmail.com}}\and\IEEEauthorblockN{Don Towsley}\IEEEauthorblockA{College of Information and Computer Sciences\\University of Massachusetts, Amherst\\{\tt towsley@cs.umass.edu}}
}

\date{}
\begin{document}
\maketitle

\begin{abstract}
	Graph states form an important class of multipartite entangled quantum states. We propose a new approach for distributing graph states across a quantum network. We consider a quantum network consisting of nodes---quantum computers within which local operations are free---and EPR pairs shared between nodes that can continually be generated. We prove upper bounds for our approach on the number of EPR pairs consumed, completion time, and amount of classical communication required, all of which are equal to or better than that of prior work\cite{distrStates}. We also reduce the problem of minimizing the completion time to distribute a graph state using our approach to a network flow problem having polynomial time complexity.
\end{abstract}

\section{Introduction}

\subsection{Motivation}

Graph states form an important class of multipartite entangled states. They are interesting both theoretically, for their importance in one-way and measurement-based quantum computing\cite{measurementCalc}\cite{PhysRevLett.86.5188}, and practically, for their applications such as to quantum metrology\cite{metrology} and secure multi-party computation\cite{entanglementApplications}.

The role of graph states in measurement-based quantum computation makes them especially interesting as a resource to distribute across a quantum network. A classic result states that any quantum computation can be done in a ``one-way'' fashion\cite{PhysRevLett.86.5188} by preparing a graph state among a set of qubits, then performing measurements and single-qubit operations based on the measurement results. Preparing such graph states among qubits in different network nodes allows the network to perform these one-way computations in a distributed manner, which can be especially useful if different network nodes receive different parts of the input to some quantum computation. This establishes distribution of graph states across a quantum network as an important service for it to provide.
\subsection{Prior Work}

There has been considerable work on the construction of graph states at a single node~\cite{campbell2007distributed}, \cite{matsuzaki2010probabilistic} and in the context of photonic cluster computing~\cite{MsmntQCBriegel09}. Much of the work on generation and distribution of graph states across a quantum network has focused on providing robustness and resilience to noise in the channels between network nodes, and memories and gates in the nodes. For example Cuquet and Calsamiglia\cite{noisyDistr} consider a similar graph state distribution protocol to ours. However, they focus on a network with a star topology as opposed to a network having an arbitrary topology as in our work. They optimize for both fidelity and fidelity decay rate given the constraint of noisy channels, instead of optimizing for EPR pair consumption and time required to distribute the graph state given the network topology. 

Pirker and D\"ur~\cite{modularArchitectures} \cite{networkStackDistr} consider graph state distribution protocols for more general network topologies like our work.  Their focus is on their placement in a larger network protocol stack and how it can be modified to work within an unreliable network rather than their performance in terms of resource requirements and time to complete entanglement distribution. 

The work of Meignant et al.~\cite{distrStates} is closest to ours in spirit. They proposed an algorithm for constructing an {\em edge-decorated complete graph (EDCG)} across a network, which is then transformed into a desired graph state. They then derived upper bounds on the number of EPR pairs consumed and on time to complete the construction, under the assumption that channels, memories, and logic is perfect, and that Bell state measurements are deterministic. We conduct a similar analysis of a different algorithm to generate and distribute graph states across a network. Our approach consists of constructing the graph state at one node and transporting the qubits to the appropriate nodes within the network. Henceforth, we refer to the algorithm proposed in~\cite{distrStates} as the EDCG algorithm.

\subsection{Overview of Results}

In this work, we study the following {\em Graph State Transfer (GST)} algorithm for distributing graph states across a quantum network. Suppose a set of network nodes desires to share a specific graph state, with one qubit from the graph state in each network node. The idea behind our algorithm is to first create a local copy of the desired graph state at one node of the network and then distribute the graph state to the relevant set of nodes.  This can be thought of as an extension of the bipartite A protocol from \cite{noisyDistr} to the general network setting. Besides introducing this generalization we make the following contributions: 
\begin{itemize}
    \item We analyze the EPR pair consumption of our GST algorithm through the derivation of an upper bound as a function of the quantum network size.  We also show that the GST algorithm never consumes more EPR pairs than the EDCG algorithm. For some networks, such as those with a binary tree topology, the difference in the numbers of EPR pairs consumed can be significant.
    \item We derive an upper bound on the time needed to distribute the graph state (henceforth referred to as {\em completion time}). We present a polynomial time algorithm that chooses paths in a network that minimizes that completion time for the GST algorithm.  We also show that the completion time of the GST algorithm is never more than that of the EDCG algorithm.
    \item We analyze the quantum memory and classical communication requirements for the GST algorithm.  The memory requirements are shown to be considerably less for the GST algorithm and the classical communication overheads are comparable.
\end{itemize}

Table \ref{comparisonTable} details these comparisons.

Our approach to distributing graph states naturally leads to a new \emph{resource graph state}: a graph state that can be distributed among a set of network nodes ahead of time that allows instantaneous distribution of any other graph state among those nodes by consuming the resource graph state, once that other graph state is known. This is useful if one knows the set of nodes that will request to share a graph in the future, but one does not yet know the exact graph state that will be requested. Our resource graph state requires maintaining fewer qubits than that of prior work.
\section{Background}

\subsection{Graph States}

A graph state\cite{multipartyEntanglement} is a type of multiple-qubit state that is useful for certain quantum computing operations between multiple parties. We represent a graph state as a graph $G=(V, E)$ where the vertices correspond to qubits. The graph state for $G$ is initiated with all qubits in the $\ket{+}$ state followed by the application of controlled $Z$ operations to all pairs of qubits corresponding to pairs of vertices in $E$. More precisely, the graph state corresponding to $G$ is
\begin{align*}
\ket{\psi_G}=\left(\prod_{(u, v)\in E}CZ_{u, v}\right)\ket{+}^{|V|}.
\end{align*}
Note that $CZ$ operations commute, so we can apply the $CZ$ operations in any order we want (or all at once).

The graph state class of multiparty entangled states is useful because, among other reasons, there is a set of quantum operations that affect the state (up to local correction operations) graphically---ie, we can think about simple, familiar graph operations instead of quantum operators and measurements. We use the following quantum operations (and their corresponding graphical operations) in this paper:
\begin{itemize}
	\item \textbf{Local complementation} of a vertex $a\in V$ replaces the subgraph corresponding to the neighbors of $a$ with its complement. This operation requires $O(|N_a|)$ bits of classical communication ($O(1)$ communication between $a$ and each of its neighbors), where $N_a$ is the set of neighbors of $a$. The quantum operations required to perform this graphical operation are given in \cite{multipartyEntanglement}.
	\item \textbf{Edge addition/deletion} of an edge $(u, v)$ creates an edge if one does not exist, or deletes it if it does. It corresponds to the $CZ_{u,v}$ operation.
	\item \textbf{$Z$-measurement} of a vertex $a$ deletes $a$ and all of its incident edges.
	\item \textbf{$Y$-measurement} of a vertex $a$ has the effect of deleting vertex $a$ and all of its incident edges, and locally complementing its neighbors. This operation requires $O(|N_a|)$ bits of classical communication: $O(1)$ communication between $a$ and each of its neighbors.
\end{itemize}
A useful property of the edge addition/deletion and $Y$-measurement operations (along with the local correction operations implicit to $Y$-measurement) is that any sequence of edge addition/deletion operations and $Y$-measurement operations can be rewritten into an equivalent sequence of operations such that the edge additions/deletions occur first and all measurements come next. All local correction operations (one per qubit) can be executed concurrently\cite{measurementCalc} at the end. This allows us to perform a sequence of $O(n)$ edge creation and $Y$-measurement operations in $O(1)$ time.

\subsection{Quantum Networks}

A quantum network is a set of nodes and edges $(V', E')$. Nodes correspond to routers and repeaters; they are computers with unlimited numbers of qubits, the capability to perform local operations and communicate within their neighborhood in order to effect long distance entanglement. An edge represents a pair of nodes connected by a quantum channel that can generate EPR pairs between them, and can regenerate EPR pairs as necessary.  The state of a quantum network at any point in time is a graph state among all the qubits in all the nodes of the network. We give an example quantum network in Figure \ref{exNetwork}.

A natural task is to distribute a graph state across a quantum network to a set of nodes. This means we alter the network state such that each qubit in a graph state exists in a specific node. Rather than preparing a graph state from some other graph state among a specified set of qubits via graphical operations, which is not always possible (in fact, it is NP-complete to determine whether transforming one graph state into another is possible\cite{npCompleteTransformation}), we only require that each qubit in the graph state be part of a specified node. To achieve this, we can use local operations within nodes (which are free in our model), and link-level EPR pair regeneration. EPR pair regeneration, an operation on qubits in different nodes, is expensive---EPR pair consumption is one of the performance metrics of a quantum networking algorithm in this model.

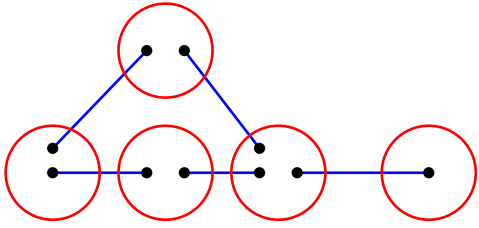
\begin{figure}
	\begin{center}
		\begin{tikzpicture}[scale=0.5]
		\coordinate (A) at (2.5 , 2.75);
		\coordinate (B) at (0, 2.75);
		\coordinate (C) at (3.5, 2.75);
		\coordinate (D) at (5.5, 2.75);
		\coordinate (E) at (0, 3.4);
		\coordinate (F) at (2.5, 6 );
		\coordinate (G) at (5.5, 3.4);
		\coordinate (H) at (3.5, 6);
		\coordinate (I) at (6.5, 2.75);
		\coordinate (J) at (10, 2.75);

		\draw[edgeStyle] (A) -- (B);
		\draw[edgeStyle] (C) -- (D);
		\draw[edgeStyle] (E) -- (F);
		\draw[edgeStyle] (G) -- (H);
		\draw[edgeStyle] (I) -- (J);
		
		\draw[nodeStyle] (3, 2.75) circle [radius=1.25];
		\draw[nodeStyle] (3, 6) circle [radius=1.25];
		\draw[nodeStyle] (0, 2.75) circle [radius=1.25];
		\draw[nodeStyle] (6, 2.75) circle [radius=1.25];
		\draw[nodeStyle] (J) circle [radius=1.25];
		
		\foreach \n in {A,B,C,D,E,F,G,H,I,J}
		\node at (\n)[circle,fill]{};
		\end{tikzpicture}
		\caption{An example quantum network. Red circles represent nodes; blue edges represent connections between nodes, which can be regenerated after being consumed by quantum operations within nodes.}
		\label{exNetwork}
	\end{center}
\end{figure}

For the problem of distributing an arbitrary graph state among a network with $n$ nodes, Meignant et al.\cite{distrStates} give an algorithm that consumes at most $\frac{n(n-1)}{2}$ EPR pairs and has a completion time of at most $n-1$ timesteps. They also propose a ``resource graph state'' (see Figure \ref{edgeDec}) that can be distributed among a network ahead of time in order to enable instantaneous distribution of an arbitrary graph state. Their resource graph state requires $\frac{n(n-1)}{2}$ qubits. We present a graph state distribution algorithm that uses at most $\frac{n(3n/2-1)}{4}$ EPR pairs. This algorithm naturally leads to an alternate resource graph state that requires $2(n-1)$ qubits.

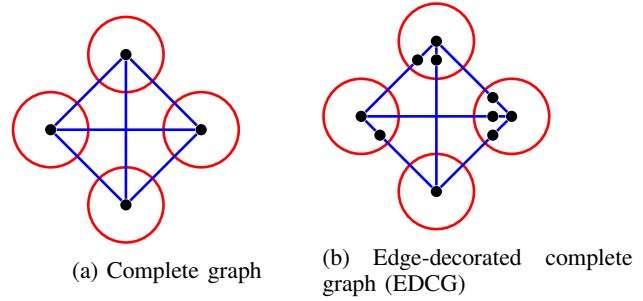
\begin{figure}
		\begin{subfigure}{.25\textwidth}
			\begin{tikzpicture}
			\node (n1) at (1, 0) [circle, fill] {};
			\node (n2) at (0, 1) [circle, fill] {};
			\node (n3) at (-1, 0) [circle, fill] {};
			\node (n4) at (0, -1) [circle, fill] {};
			
			\draw[nodeStyle] (n1) circle [radius=0.5];
			\draw[nodeStyle] (n2) circle [radius=0.5];
			\draw[nodeStyle] (n3) circle [radius=0.5];
			\draw[nodeStyle] (n4) circle [radius=0.5];
			
			\draw[edgeStyle] (n1) -- (n2);
			\draw[edgeStyle] (n1) -- (n3);
			\draw[edgeStyle] (n1) -- (n4);
			\draw[edgeStyle] (n2) -- (n3);
			\draw[edgeStyle] (n2) -- (n4);
			\draw[edgeStyle] (n3) -- (n4);
			\end{tikzpicture}
			
			\caption{Complete graph}
		\end{subfigure}\begin{subfigure}{.25\textwidth}
			\begin{tikzpicture}
			\node (v1) at (1, 0) [circle, fill] {};
			\node (v2) at (0, 1) [circle, fill] {};
			\node (v3) at (-1, 0) [circle, fill] {};
			\node (v4) at (0, -1) [circle, fill] {};
			
			\draw[nodeStyle] (n1) circle [radius=0.5];
			\draw[nodeStyle] (n2) circle [radius=0.5];
			\draw[nodeStyle] (n3) circle [radius=0.5];
			\draw[nodeStyle] (n4) circle [radius=0.5];
			
			\node (e12) at (.75, .25) [circle, fill] {};
			\node (e13) at (.75, 0) [circle, fill] {};
			\node (e14) at (.75, -.25) [circle, fill] {};
			\node (e23) at (-.25, .75) [circle, fill] {};
			\node (e24) at (0, .75) [circle, fill] {};
			\node (e34) at (-.75, -.25) [circle, fill] {};
			
			\draw[edgeStyle] (v1) -- (e12); \draw[edgeStyle] (e12) -- (v2);
			\draw[edgeStyle] (v1) -- (e13); \draw[edgeStyle] (e13) -- (v3);
			\draw[edgeStyle] (v1) -- (e14); \draw[edgeStyle] (e14) -- (v4);
			\draw[edgeStyle] (v2) -- (e23); \draw[edgeStyle] (e23) -- (v3);
			\draw[edgeStyle] (v2) -- (e24); \draw[edgeStyle] (e24) -- (v4);
			\draw[edgeStyle] (v3) -- (e34); \draw[edgeStyle] (e34) -- (v4);
			\end{tikzpicture}
			
			\caption{Edge-decorated complete graph (EDCG)}
		\end{subfigure}
	\caption{A 4 node example of a resource graph state. (a) A 4-node network such that the network graph state is the complete graph among 4 qubits, each qubit in a different node. (b) A 4-node network such that the network graph state is the {\em edge-decorated complete graph}: a complete graph with additional vertices added to split each edge into two. The additional vertices added can exist in either of the nodes of the edge which that vertex split in two. $Z$ or $Y$ measuring a decoration vertex deletes or preserves the original edge, respectively. By $Z$ or $Y$ measuring each decoration vertex (along with associated local correction operations), any 4-qubit graph state can be prepared among the 4 nodes.}
	\label{edgeDec}
\end{figure}

\section{Connection Transfer} \label{LcDistr}

We start with a simple sequence of operations we refer to as {\em connection transfer}. This operation starts with a qubit $a$ at a node $A\in V'$ that is connected to other qubits, which are possibly outside $A$. $A$ also includes a second qubit $b$ that is entangled with a third qubit $c$ residing at another node. Connection transfer changes the network graph state such that the edges between qubit $a$ and its neighborhood are connected to qubit $c$ instead of $a$. See Figure \ref{connTransSetupEnd} for the setup and end result of connection transfer. We present two approaches to connection transfer: via graphical operations, and via teleportation.

\begin{figure}
	\centering
	\begin{subfigure}{.24\textwidth}
		\centering
		\begin{tikzpicture}[scale=0.5]
		\coordinate (c) at (1, 0);
		\coordinate (b) at (3.8, 0);\textit{}
		\coordinate (a) at (4.6, 0);
		\coordinate (n1) at (6.5, 0);
		\coordinate (n2) at (6.5, -1);
		\coordinate (n3) at (6.5, 1);
		
		\draw[nodeStyle] (c) circle [radius=1.1];
		\draw[nodeStyle] (4.2, 0) circle [radius=1.1];
		
		\foreach \n in {c, b, a}
		\node at (\n)[circle,fill,label=\n]{};
		
		\draw[edgeStyle] (c) -- (b);
		\draw[edgeStyle] (a) -- (n1);
		\draw[edgeStyle] (a) -- (n2);
		\draw[edgeStyle] (a) -- (n3);
		\end{tikzpicture}
		\caption{Setup.}
		\label{connTransSetupEndSetup}
	\end{subfigure}
	\begin{subfigure}{.24\textwidth}
		\begin{tikzpicture}[scale=0.5]
		
		\draw[nodeStyle] (c) circle [radius=1.1];
		\draw[nodeStyle] (4, 0) circle [radius=1.1];
		
		\foreach \n in {c}
		\node at (\n)[circle,fill,label=\n]{};
		
		\draw[edgeStyle] (c) -- (n1);
		\draw[edgeStyle] (c) -- (n2);
		\draw[edgeStyle] (c) -- (n3);
		\end{tikzpicture}
		\caption{End result.}
		\label{connTransSetupEndEnd}
	\end{subfigure}
	\caption{The setup and end result of the connection transfer process. We transfer the edges connected to $a$ to qubit $c$, by consuming the EPR pair between $b$ and $c$.}
	\label{connTransSetupEnd}
\end{figure}
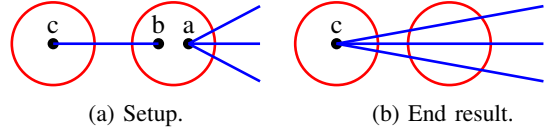

Figure \ref{connTransGraph} details connection transfer via graphical operations. First we create an edge between $a$ and $b$ with a local $CZ$ operation. Then we $Y$-measure both $a$ and $b$. The successive $Y$-measurements locally complement $a$'s neighborhood twice, but the second such local complementation undoes the first, making the net effect of the two $Y$-measurements to transfer $a$'s connections to $c$. This process consumes one (non-local) EPR pair.

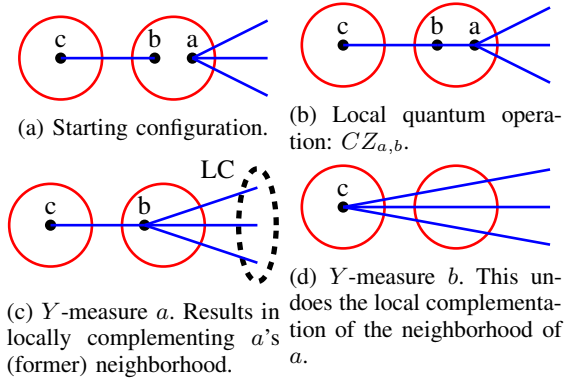
\begin{figure}
	\centering
	\begin{subfigure}{.22\textwidth}
		\centering
		\begin{tikzpicture}[scale=0.5]
		\coordinate (c) at (1, 0);
		\coordinate (b) at (3.5, 0);\textit{}
		\coordinate (a) at (4.5, 0);
		\coordinate (n1) at (6.5, 0);
		\coordinate (n2) at (6.5, -1);
		\coordinate (n3) at (6.5, 1);
		
		\draw[nodeStyle] (c) circle [radius=1.1];
		\draw[nodeStyle] (4, 0) circle [radius=1.1];
		
		\foreach \n in {c, b, a}
		\node at (\n)[circle,fill,label=\n]{};
		
		\draw[edgeStyle] (c) -- (b);
		\draw[edgeStyle] (a) -- (n1);
		\draw[edgeStyle] (a) -- (n2);
		\draw[edgeStyle] (a) -- (n3);
		\end{tikzpicture}
		\caption{Starting configuration.}
		\label{connTransStart}
	\end{subfigure}
	\begin{subfigure}{.22\textwidth}
		\centering
		\begin{tikzpicture}[scale=0.5]
		
		\draw[nodeStyle] (c) circle [radius=1.1];
		\draw[nodeStyle] (4, 0) circle [radius=1.1];
		
		\foreach \n in {c, b, a}
		\node at (\n)[circle,fill,label=\n]{};
		
		\draw[edgeStyle] (c) -- (b);
		\draw[edgeStyle] (a) -- (b);
		\draw[edgeStyle] (a) -- (n1);
		\draw[edgeStyle] (a) -- (n2);
		\draw[edgeStyle] (a) -- (n3);
		\end{tikzpicture}
		\caption{Local quantum operation: $CZ_{a,b}$.}
	\end{subfigure}
	\begin{subfigure}{.22\textwidth}
		\centering
		\begin{tikzpicture}[scale=0.5]
		
		\draw[nodeStyle] (c) circle [radius=1.1];
		\draw[nodeStyle] (4, 0) circle [radius=1.1];
		
		\draw [dashed, line width=2pt] (n1) ellipse (.5 and 1.5);
		\node at (n3)[label={[xshift=-15]LC}]{};
		
		\foreach \n in {c, b}
		\node at (\n)[circle,fill,label=\n]{};
		
		\draw[edgeStyle] (c) -- (b);
		\draw[edgeStyle] (b) -- (n1);
		\draw[edgeStyle] (b) -- (n2);
		\draw[edgeStyle] (b) -- (n3);
		\end{tikzpicture}
		\caption{$Y$-measure $a$. Results in locally complementing $a$'s (former) neighborhood.}
	\end{subfigure}
	\begin{subfigure}{.22\textwidth}
		\centering
		\begin{tikzpicture}[scale=0.5]
		
		\draw[nodeStyle] (c) circle [radius=1.1];
		\draw[nodeStyle] (4, 0) circle [radius=1.1];
		
		\foreach \n in {c}
		\node at (\n)[circle,fill,label=\n]{};
		
		\draw[edgeStyle] (c) -- (n1);
		\draw[edgeStyle] (c) -- (n2);
		\draw[edgeStyle] (c) -- (n3);
		\end{tikzpicture}
		\caption{$Y$-measure $b$. This undoes the local complementation of the neighborhood of $a$.}
		\label{connTransEnd}
	\end{subfigure}
	\caption{Connection transfer via graphical operations.}
	\label{connTransGraph}
\end{figure}

Connection transfer via teleportation is straightforward. Again, we start with a qubit $a$ whose edges we wish to transfer to a qubit $c$. Qubit $c$ is connected to a qubit $b$ located in the same node as $a$. This situation is depicted in Figure \ref{connTransSetupEnd}(a). The initial state is
\begin{align*}
\ket{\psi_G}
=&\frac{1}{\sqrt{2}}\big(\ket{+}_b\ket{0}_c+\ket{-}_b\ket{1}_c\big) \\
&\frac{1}{\sqrt{2}}\Big(\ket{0}_a+\ket{1}_a\prod_{v\in N_a}Z_v\Big) \\
&\Big(\prod_{(u, v)\in E''}CZ_{u,v}\Big)\ket{+}^{\otimes \left|V\setminus \{a, b, c\}\right|}
\end{align*}
where $E''$ is the edge set of the network's graph state except for those incident to $a$, and except for $(b, c)$. We break this expression down term by term. The $\ket{+}^{\otimes \left|V\setminus \{a, b, c\}\right|}$ term corresponds to all qubits except $a$, $b$, and $c$ prepared in the $\ket{+}$ state. The $\prod_{(u, v)\in E''}CZ_{u,v}$ operations create all the edges except for $(b,c)$ and those connected to $a$. The $\frac{1}{\sqrt{2}}\left(\ket{0}_a+\ket{1}_a\prod_{v\in N_a}Z_v\right)$ term creates the qubit $a$ and the edges between $a$ and the vertices in its neighborhood $N_a$. The $\frac{1}{\sqrt{2}}\left(\ket{+}_b\ket{0}_c+\ket{-}_b\ket{1}_c\right)$ term creates the qubits $b$ and $c$ and the edge between them.

It is easy to see that measuring qubits $a$ and $b$ in the basis
\begin{equation}
\Bigg\{
\begin{aligned}
&\frac{1}{\sqrt{2}}\left(\ket{0}_a\ket{+}_b\pm\ket{1}_a\ket{-}_b\right), \\ &\frac{1}{\sqrt{2}}\left(\ket{0}_a\ket{-}_b\pm\ket{1}_a\ket{+}_b\right)
\end{aligned}
\Bigg\}\label{teleBasis}
\end{equation}
results in the desired transfer of $a$'s connections to $c$. To see this, consider what happens when we obtain the measurement result $\ket{\phi}=\frac{1}{\sqrt{2}}\left(\ket{0}_a\ket{+}_b+\ket{1}_a\ket{-}_b\right)$:
\begin{align*}
\braket{\phi}{\psi_G}=&\frac{1}{2\sqrt{2}}\Big(\ket{0}_c+\ket{1}_c\prod_{v\in N_a}Z_v\Big) \\
&\Big(\prod_{(u, v)\in E'}CZ_{u,v}\Big)\ket{+}^{\otimes \left|V\setminus \{a, b, c\}\right|}.
\end{align*}
This is precisely the graph state depicted in Figure \ref{connTransSetupEnd}(b). If the measurement result is another Bell state besides $\ket{\phi}$ then an $X$ and/or $Z$ gate correction will also need to be applied to qubit $c$, as is done with conventional quantum teleportation of one qubit.

Note that the graphical and teleportation approaches to connection transfer are essentially equivalent---the local $CZ$ operation of the graphical approach effects a change of basis, allowing the 2 single-qubit $Y$-measurements to achieve the same effect as the multi-qubit measurement in the basis \eqref{teleBasis} used in teleportation. The graphical approach requirement that each node be able to perform local $CZ$ operations and $Y$-measurements is no stricter than the operation requirements of nodes considered in prior work\cite{distrStates}.

\section{Graph State Distribution}

We can transfer a qubit's connections to a node not connected to that qubit's node through a sequence of connection transfers along a path of edges in the network running from the starting node to the desired node. This suggests the following algorithm for generating a graph state. First, generate a local copy of the graph state at some node via local $CZ$ operations, which are free in our model. Next, transfer the connections (using either of the aforementioned connection transfer methods) of each qubit to its corresponding node. Figure \ref{beginningEnd} illustrates the starting state and end result of this algorithm for an example network and desired final graph state. We call this algorithm the Graph State Transfer algorithm, or the \textbf{GST algorithm}.

This requires the root node to maintain $|S|$ qubits (where $S$ is the set of network nodes that will share the final graph state) and prepare them in an entangled state via local $CZ$ operations. This may be a difficult requirement to meet for large networks; however, it is also required of the graph state distribution approach in \cite{distrStates}.

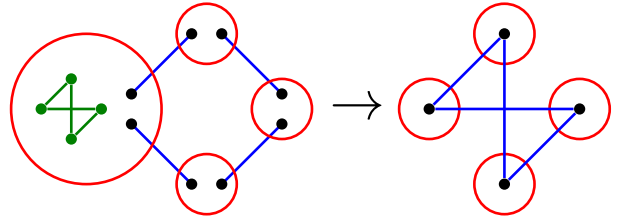
\begin{figure}
	\begin{center}
		\begin{tikzpicture}[scale=0.4]
		\coordinate (A) at (2, 0);
		\coordinate (B) at (0, 2);
		\coordinate (C) at (3, 0);
		\coordinate (D) at (5, 2);
		\coordinate (E) at (0, 3);
		\coordinate (F) at (2, 5);
		\coordinate (G) at (5, 3);
		\coordinate (H) at (3, 5);
		
		\draw[edgeStyle] (A) -- (B);
		\draw[edgeStyle] (C) -- (D);
		\draw[edgeStyle] (E) -- (F);
		\draw[edgeStyle] (G) -- (H);
		
		\draw[nodeStyle] (2.5, 0) circle [radius=1];
		\draw[nodeStyle] (2.5, 5) circle [radius=1];
		\draw[nodeStyle] (-1.5, 2.5) circle [radius=2.5];
		\draw[nodeStyle] (5, 2.5) circle [radius=1];
		
		\foreach \n in {A,B,C,D,E,F,G,H}
		\node at (\n)[circle,fill]{};
		
		\node (n1) at (-3, 2.5)[circle,fill,localCopyColor]{};
		\node (n2) at (-2, 3.5)[circle,fill,localCopyColor]{};
		\node (n3) at (-2, 1.5)[circle,fill,localCopyColor]{};
		\node (n4) at (-1, 2.5)[circle,fill,localCopyColor]{};
		
		\draw[localCopyStyle] (n1) -- (n2);
		\draw[localCopyStyle] (n2) -- (n3);
		\draw[localCopyStyle] (n3) -- (n4);
		\draw[localCopyStyle] (n1) -- (n4);
		
		\node at (7.5, 2.5) {\huge$\to$};
		\end{tikzpicture}
		\begin{tikzpicture}[scale=0.4]
		\node (n1) at (0, 2.5)[circle,fill]{};
		\draw[nodeStyle] (n1) circle [radius=1];
		\node (n2) at (2.5, 5)[circle,fill]{};
		\draw[nodeStyle] (n2) circle [radius=1];
		\node (n3) at (5, 2.5)[circle,fill]{};
		\draw[nodeStyle] (n3) circle [radius=1];
		\node (n4) at (2.5, 0)[circle,fill]{};
		\draw[nodeStyle] (n4) circle [radius=1];
		
		\draw[edgeStyle] (n1) -- (n2);
		\draw[edgeStyle] (n3) -- (n4);
		\draw[edgeStyle] (n1) -- (n3);
		\draw[edgeStyle] (n2) -- (n4);
		\end{tikzpicture}
	\end{center}
	\caption{Example setup and end result of our GST algorithm. A local copy of the final graph state (green) is prepared within a node and distributed throughout the network.}\label{beginningEnd}
\end{figure}

\subsection{Resource Graph State}

This approach to distributing graph states suggests a resource graph state---a graph state that can be distributed among a set of nodes ahead of time that allows any arbitrary graph state to be distributed among those nodes in one timestep. Resource graph states are useful if we know ahead of time that a set of nodes (or a superset of nodes) will request a graph state, but we do not know what that graph state will be. We choose one node in the network (called the ``root node'') and have the network graph state be such that the root node shares an entangled pair with every other node that will share the desired final graph state, as in Figure \ref{resourceGraph}. This allows us to generate an arbitrary graph state in one time step by generating the local copy at the root and distributing the graph state as usual, using either connection transfer method.

This resource graph state requires $2(n-1)$ qubits, where $n$ is the number of nodes that will share the graph state. This is an improvement over the $\frac{n(n-1)}{2}$ qubits needed for the EDCG (Figure \ref{edgeDec}), the resource graph state proposed in \cite{distrStates}. This improvement is significant because long-term maintenance of memory qubits is, and likely will continue to be, a challenging engineering problem.

\begin{figure}
	\begin{center}
		\begin{tikzpicture}[scale=1]
		\node (c0) at (-0.5, 0.5)[circle,fill]{};
		\node (c1) at (-0.25, 0.5)[circle,fill]{};
		\node (c2) at (0, 0.5)[circle,fill]{};
		\node (c4) at (0.5, 0.5)[circle,fill]{};
		
		\node (n0) at (-2, 2)[circle,fill]{};
		\node (n1) at (-1, 2)[circle,fill]{};
		\node (n2) at (0, 2)[circle,fill]{};
		\node at (1, 2){$\cdots$};
		\node (nn) at (2, 2)[circle,fill]{};
		
		\draw[edgeStyle] (n0) -- (c0);
		\draw[edgeStyle] (n1) -- (c1);
		\draw[edgeStyle] (n2) -- (c2);
		\draw[edgeStyle] (nn) -- (c4);
		
		\foreach \n in {n0, n1, n2, nn}
		\draw[nodeStyle] (\n) circle [radius=.3];
		
		\draw[nodeStyle] (0, .5) circle [radius=1];
		\node at (1.25, -0.25) {root};
		\node at (0, 2.25) [label={[align=center]other nodes that will\\ share the graph state}] {};
		\end{tikzpicture}
	\end{center}
	\caption{A resource graph state that requires $2(n-1)$ qubits, where $n$ is number of nodes that will share the final graph state.}
	\label{resourceGraph}
\end{figure}
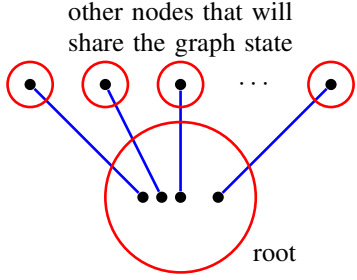

\section{EPR Pair Consumption}\label{eprSection}

Each connection transfer operation consumes one EPR pair. For each qubit in a graph state whose connections are transferred to a relevant node, those connection transfers consume a number of EPR pairs equal to the length of the path in the network from the root node the relevant node. Thus the total number of EPR pairs used to distribute a graph state across a network depends on the choice of paths from the root node to every other node in the network (and also implicitly depends on the choice of a root node). The number of EPR pairs consumed equals the sum of the lengths of such paths.

Upper bounding the number of EPR pairs consumed by this algorithm when distributing a graph state among a set of nodes $S$ is thus equivalent to upper bounding the sum of minimum path lengths from some root node to every node in $S$. We upper bound this sum by choosing the root node to our advantage. For any connected graph with $n$ vertices and any vertex $v$, at most $n-i$ vertices can be distance $i$ away from $v$. This means the sum of minimum path lengths from $v$ to all vertices in $S$, which we call $Ne(S)$, is at most
\begin{align*}
\begin{split}
N_e(S)\leq\mbox{}&(n-1)+(n-2)+\cdots+(n-|S|) \\
=\mbox{}&\frac{|S|(2n-|S|-1)}{2}.
\end{split}
\end{align*}
In particular, if $S=V'$ (ie. we are distributing a graph state across the entire network) then we use at most
\begin{align}
\begin{split}
N_e(V')\leq\mbox{}&(n-1)+(n-2)+\cdots+1\\
=\mbox{}& \frac{n(n-1)}{2}
\end{split}\label{dumbEprBound}
\end{align}
EPR pairs. Note that this upper bound is achieved when the network is a linear graph with the root at one end of the line.

Suppose we have the flexibility to select any vertex as the root. For any connected graph with $n$ vertices and maximum degree of at least two, basic graph theory tells us that there exists a vertex $v$ that is distance at most $\left\lceil\frac{n-1}{2}\right\rceil$ from any other vertex (see eg. \cite{graphRadiusBounds} Theorem 4.1). Thus by choosing a root in the center of the graph, we can replace every term in the above sum that is at least $\left\lceil\frac{n-1}{2}\right\rceil$ by $\left\lceil\frac{n-1}{2}\right\rceil$ to get a better upper bound. We compute this bound for even and odd $n$.

For even $n$, we replace every term in the sum from \eqref{dumbEprBound} that is at least $\frac{n}{2}$ by $\frac{n}{2}$ to get
\begin{align}
N_e(V')\leq\mbox{} &\frac{n}{2}\cdot\frac{n}{2}+\left(\frac{n}{2}-1\right)+\left(\frac{n}{2}-2\right)+\cdots+1 \nonumber \\
=\mbox{}&\frac{n^2}{4}+\frac{\frac{n}{2}\left(\frac{n}{2}-1\right)}{2} \nonumber \\
=\mbox{}&\frac{3n^2-2n}{8}.\label{evenEprBound}
\end{align}
For odd $n$, we replace every term in the sum from \eqref{dumbEprBound} that is at least $\frac{n-1}{2}$ by $\frac{n-1}{2}$ to get
\begin{alignat}{3}
 N_e(V')\leq\mbox{} & \frac{n+1}{2}\cdot\frac{n-1}{2}  +\left(\frac{n-1}{2}-1\right) \nonumber \\
 &  \quad\quad+\left(\frac{n-1}{2}-2\right)   +\cdots+1 \nonumber \\
 =\mbox{}& \frac{n^2-1}{4}+\frac{\frac{n-1}{2}\left(\frac{n-1}{2}-1\right)}{2} \nonumber \\
 =\mbox{}& \frac{3n^2-4n+1}{8}. \label{oddEprBound}
\end{alignat}
As $n> 0$, the even $n$ bound from \eqref{evenEprBound} is greater than the odd $n$ bound from \eqref{oddEprBound}, so in general $N_e(V')\leq(3n^2-2n)/8$.

This EPR pair consumption upper bound is lower than that for the EDCG algorithm, \cite{distrStates}, which uses up to $\frac{n(n-1)}{2}$ EPR pairs. However, we can also prove a stronger result: that we always use less than or an equal number of EPR pairs used by the EDCG algorithm.

The EDCG algorithm creates an edge-decorated complete graph (EDCG) between the set of nodes $S$ that will share the final graph state (see Figure \ref{edgeDec}). To create an EDCG among $S=\{s_1, s_2, \cdots, s_m\}$, the EDCG algorithm creates an $m$-qubit GHZ state between $\{s_1, \cdots, s_m\}$, then an $m-1$ qubit GHZ state between $\{s_2, \cdots, s_m\}$, then an $m-2$ qubit GHZ state between $\{s_3, \cdots, s_m\}$, etc, until creating a 2 qubit GHZ state (ie, just an EPR pair/edge in the network's graph state) between $s_{m-1}$ and $s_m$. These GHZ states are then combined via local operations at each node to form an EDCG. Then, each edge in the complete graph is either deleted or kept, by $Z$-measuring or $Y$-measuring the decoration vertex on that edge, respectively. See \cite{distrStates} Figures 7 and 8 for more detail.

\begin{theorem}\label{eprTheorem}
    The GST algorithm always uses less than or an equal number of EPR pairs used by the EDCG algorithm.
\end{theorem}
\begin{proof}
Creating a GHZ state between $\{s_k,\cdots,s_m\}$ requires at least as many EPR pairs as performing connection transfer from $s_m$ to $s_k$, as the former will use EPR pairs along a path in the network between $s_k$ and $s_m$ (and possibly more EPR pairs). Thus, performing connection transfer $|S|-1$ times between $s_m$ and the other nodes in $S$ uses no more EPR pairs than generating all the GHZ states required by the EDCG algorithm. So by choosing $s_m$ as our root node, we use no more EPR pairs than the EDCG algorithm does to distribute a graph state among nodes in $S$.
\end{proof}

\subsection{An Example: Full Binary Tree}\label{eprExampleTree}

Here we provide an example where there is a large gap in the numbers of EPR pairs consumed by the GST algorithm and by the EDCG algorithm to distribute a graph state among every node in a full binary tree. Consider a full binary tree of height $h$ (by convention we consider the trivial tree with 1 vertex to have height 0); this graph has $n=2^{h+1}-1$ vertices. We choose the root of the tree as the root node of the network, and the paths we use to transfer connections to every other node of the network are obvious as there is only one choice of path for each node since the network structure is a tree. The number of EPR pairs consumed is the sum of path lengths from the root node to every other node:
\begin{equation*}
    \sum_{i=1}^h i2^i=2^{h+1}(h-1) + 2=\Theta(n\log n).
\end{equation*}

It is a straightforward calculation to show that the EDCG algorithm requires $\frac{n(n-1)}{2}$ EPR pairs to distribute a graph state to every node in this network with the EDCG algorithm. For this example, the EDCG algorithm consumes $\Theta (n/\log n)$ more EPR pairs. 

\section{Minimizing Completion Time}\label{timeOptSec}

\subsection{Parallelization}
First, we show that any sequence of connection transfers that does not require using any network edge more than once can be done simultaneously in one timestep.

For the graphical connection transfer approach, note that any sequence of connection transfers is a sequence of controlled-$Z$ operations, $Y$ measurements (at most one per qubit), and local correction operations required by the measurement results. We can rearrange those operations \cite{measurementCalc} (see Section 5.2 of \cite{measurementCalc} for details) into a different sequence of operations with the same effect such that the new sequence of operations consists first of controlled-$Z$ operations, followed by measurements, and then local correction operations. The controlled-$Z$ operations will be done on the same qubit pairs as the original sequence of operations, and the measurements and local correction operations will also be done on the same qubits as the original sequence.

This means any sequence of connection transfer operations that does not use any edge in the network more than once can equivalently be done as a sequence of (in order):
\begin{enumerate}
	\item Controlled-$Z$ operations. These can all be done at once as these operations commute.
	\item Measurements. These can all be done at once because the measurements are of different qubits.
	\item Local correction operations based on the measurement results.
\end{enumerate}

The teleportation approach to connection transfer, like the graphical approach, also allows connection transfers among distinct edges in the network to be parallelized. We can perform the necessary local correction operations ($X$ and/or $Z$ gates) on the final qubit in the connection transfer path based on the results of all the measurements in the connection transfer path. The exact correction operations on the final qubit are the correction operations that would have been done on the qubits in the path, done in reverse order of their appearance in the path.

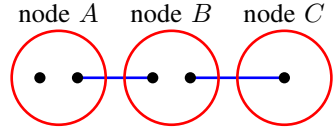
\begin{figure}
	\begin{center}
		\begin{tikzpicture}[scale=0.5]
		\coordinate (A) at (-0.5, 2.75);
		\coordinate (B) at (0.5, 2.75);
		\coordinate (C) at (2.5, 2.75);
		\coordinate (D) at (3.5, 2.75);
		\coordinate (E) at (6, 2.75);

		\draw[edgeStyle] (B) -- (C);
		\draw[edgeStyle] (D) -- (E);
		
		\draw[nodeStyle] (0, 2.75) circle [radius=1.25];
		\draw[nodeStyle] (3, 2.75) circle [radius=1.25];
		\draw[nodeStyle] (6, 2.75) circle [radius=1.25];
		
		\foreach \n in {A,B,C,D,E}
		\node at (\n)[circle,fill]{};
		
		\node at (0, 4) [label={[align=center]node $A$}] {};
		\node at (3, 4) [label={[align=center]node $B$}] {};
		\node at (6, 4) [label={[align=center]node $C$}] {};
		\end{tikzpicture}
		\caption{The setup for teleportation along two edges in the network. The pair of qubits in node $A$ and the pair of qubits in node $B$ will be measured in the basis in Equation \eqref{teleBasis}.}\label{twoTeleportations}
	\end{center}
\end{figure}

To illustrate this rearrangement of operations from multiple teleportations, we show exactly how it works for the case of two teleportations in a row; see Figure \ref{twoTeleportations} for the setup. When teleporting a state from node $A$ to node $B$ and then from node $B$ to node $C$, the usual sequence of operations is:
\begin{enumerate}
    \item Measure two qubits in node $A$ in the basis in Equation \eqref{teleBasis}.
    \item Based on the measurement result, perform a correction operation ($X$ and/or $Z$ gates) on a qubit in node $B$.
    \item Measure two qubits in node $B$ in the basis in Equation \eqref{teleBasis}.
    \item Based on the measurement result, perform a correction operation ($X$ and/or $Z$ gates) on a qubit in node $C$.
\end{enumerate}
Note that instead of correcting for the first measurement result in step 2, we can teleport the uncorrected state from node $B$ to node $C$ and then perform the correction operations that we would have done in step 2 on the qubit in node $C$. This results in the sequence of operations:
\begin{enumerate}
    \item Measure two qubits in node $A$ in the basis in Equation \eqref{teleBasis}.
    \item Measure two qubits in node $B$ in the basis in Equation \eqref{teleBasis}.
    \item Based on the second measurement result, perform a correction operation ($X$ and/or $Z$ gates) on the qubit in node $C$.
    \item Based on the first measurement result, perform a correction operation ($X$ and/or $Z$ gates) on the qubit in node $C$.
\end{enumerate}
We can do both measurement operations at once since they are performed on different qubits. The measurement results can then be reported to node $c$, followed by all local correction operations on node $c$.

This also easily generalizes to sequences of more than two teleportations; we just continue teleporting uncorrected states to the last qubit in the path and then do all correction operations at that last qubit. This means the only operations done at each node in a path of teleportation connection transfers (except the last node in the path) are the measurement operations, which can all be done at once because they are measurements on different qubits. Also, because the only local correction operations are performed at the end of any chain of connection transfers, the measurement results need only be communicated to the last node of any path.

Even though the final qubit in a chain of $n$ teleportations may require up to $2n$ correction operations, that sequence of operations will be equivalent to one of the 16 elements of the Pauli group on 1 qubit. So the up to $2n$ correction operations required can be accomplished in $O(1)$ time by applying the relevant element from the Pauli group.

We refer to the time it takes to perform simultaneous $CZ$ operations, measurements, local correction operations, and generate any EPR pairs as needed, as a timestep. Thus any sequence of connection transfer operations that does not use any edge in the network more than once, done via the graphical approach or teleportation approach, can be executed in one timestep. In general, distributing a graph state across an $n$ node network will take no more than $n-1$ timesteps. This is because the connection transfers on each path from the root node to another node can be executed in one timestep, and there are at most $n-1$ such paths---one per node that receives connections from the root node.

\subsection{Optimization via Path Selection}
We can often do better than $n-1$ timesteps. We can minimize the completion time by solving a network flow problem (see eg. \cite{clrs} chapter 26). Specifically, given a network graph, a root node, and a set of vertices $S$ of the network graph that will share the final graph state, we construct a network flow problem instance such that its maximum flow is $|S|$ iff there is a set of $|S|$ paths from the root to the vertices in $S$ such that no edge in the graph is used more than $k$ times (which allows us to distribute a graph state in $k$ timesteps). A binary search on $k$, as well as trying all possible root nodes, gives the optimal time to distribute a graph state among the nodes in $S$.

The construction is as follows. Start with the original network graph with each edge having weight $k$. Add a new vertex $t$. Finally, add edges from each vertex in $S$ to $t$ with weight 1 (see Figure \ref{networkFlow}). This network flow problem instance is related to completion time minimization by the following theorem.

\begin{theorem}
	In the network flow construction given in Figure \ref{networkFlow}, the max flow from the root node to $t$ is $|S|$ iff there exist $|S|$ paths in the network, each from the root to a different node in $S$, such that each edge is used at most $k$ times.
\end{theorem}

\begin{proof}
The $\impliedby$ direction is obvious, as we can construct a flow of value $|S|$ by adding all of the $|S|$ paths, and setting the flow of the edges from $S$ to $t$ to be one.

For the $\implies$ direction, start with a maximum flow of value $|S|$ from the root node to $t$. The flow decomposition theorem allows us to decompose a max flow of value $|S|$ into path flows (and cycle flows, which we can ignore) that combine to form the max flow. Because there are $|S|$ edges going into $t$ each with weight one, those path flows must have value one and there must be $|S|$ of them. Those $|S|$ path flows each with value one from the root node to $t$ give us $|S|$ paths from the root node to each node in $S$. Because each edge in the network flow instance has capacity at most $k$, each edge in the original network graph must be used at most $k$ times by all the paths.
\end{proof}

Our completion time minimization algorithm is given in Algorithm \ref{timeMinAlgorithm}.
See Figure \ref{exNetworkFlow} for an example of connection transfer paths found by our network flow approach.

\begin{algorithm*}
`   \caption{Our algorithm for minimizing the completion time of distributing a graph state in a network with graph structure $G$, to some subset $S$ of nodes in the network.}\label{timeMinAlgorithm}
    \begin{algorithmic}[1]
        \Procedure{NetworkFlow}{$G$, $S$, $k$, root}
            \State Start with the network graph $G$ and assign all edge weights $k$.
            \State Add a vertex $t$.
            \State Add $|S|$ edges, from every vertex in $S$ to $t$, with weight 1.
            \State Let root be the source node of this network flow instance and let $t$ be the sink node.
            \State \Return the max flow of this network.
        \EndProcedure
        \Procedure{MinimizeCompletionTime}{G, S}
            \ForAll{possible root nodes $v\in V(G)$}
            \State Use binary search on $k$ to find the minimum $k\in\{1,2,\cdots,|S|\}$ such that $\mathrm{NetworkFlow} (G, S, k, v)$ has max flow $|S|$.
            \State Let $k_v$ be this minimum $k$ value.
            \EndFor
            \State Let $\mathrm{root}=\arg\min_{v\in V(G)}k_v$.
            \State Let $k=\min_{v\in V(G)}k_v$.
            \State Use the flow decomposition theorem to extract $|S|$ paths from $\mathrm{NetworkFlow} (G, S, k, \mathrm{root})$.
            \State Use these $|S|$ paths to distribute the graph state among the network from the root node to the nodes in $S$.
        \EndProcedure
    \end{algorithmic}
\end{algorithm*}

Note that selecting connection transfer paths in the network that minimize completion time may result in more EPR pairs consumed than indicated by our previously derived upper bound. This is because the paths found from the network flow problem may not be the shortest paths from the root node to the nodes in $S$, and the EPR pair consumption bound relies on using the shortest paths in the network. Exploring the trade-off between reducing completion time and reducing EPR pair consumption is a subject for future work.

\begin{figure}
	\begin{tikzpicture}[scale=0.3]
	\coordinate (v1) at   (-8, 0);
	\coordinate (v2) at   (-4, 0);
	\coordinate (v3) at   ( 0, 0);
	\coordinate (v4) at   ( 4, 0);
	\coordinate (dots) at ( 8, 0);
	\coordinate (vn) at   ( 12, 0);
	\coordinate (t) at    ( 0, -2);
	
	\node at (v1)[circle,fill,label={$v_1$}]{};
	\node at (v2)[circle,fill,label=$v_2$]{};
	\node at (v3)[circle,fill,label=$v_3$]{};
	\node at (v4)[circle,fill,label=$v_4$]{};
	\node at (dots) {$\cdots$};
	\node at (vn)[circle,fill,label=$v_n$]{};
	\node at (t) [circle,fill,label=below:$t$]{};
	\node at (2, -3.5) [label=below:{(edges from $S$ to $t$ have weight 1)}]{};
	
	\foreach \n in {v2,v3,v4,vn}
	\draw[edgeStyle] (\n) -- (t);
	
	\draw [dashed, line width=2pt] (2, 0.5) ellipse (13 and 1.75);
	\node at (5, 2.1)[label={network graph with edge weights $k$}]{};
	\end{tikzpicture}
	\caption{Let every edge in this graph (which is the network graph, plus $|S|$ edges from every node in $S$ to an extra vertex $t$) have weight $k$, except for the edges to $t$ which have weight 1. Then the max flow from $v_1$ (the root) to $t$ is $|S|$ iff there is set of $|S|$ paths from the root to every node in $S$ such that every edge in the network is used at most $k$ times.}
	\label{networkFlow}
\end{figure}
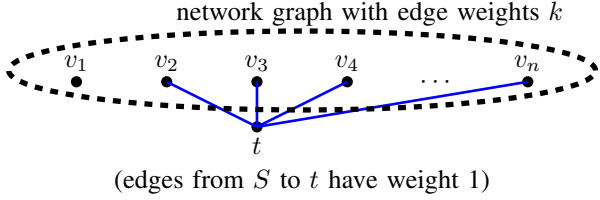

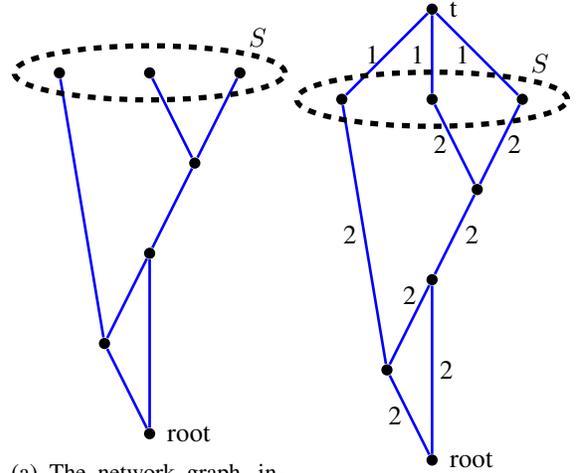
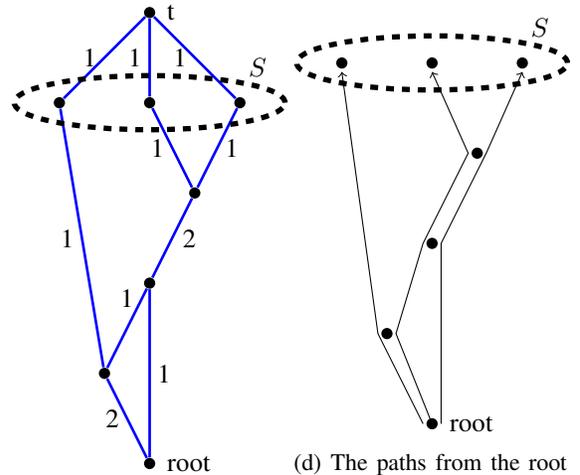
\begin{figure}
	\centering
	\begin{subfigure}{0.22\textwidth}
		\centering
		\begin{tikzpicture}[scale=1.2]
		\node (root) at (0, 0)[circle,fill]{};
		\draw (root) node [right=5pt]{root};
		\node (n1) at (-0.5, 1.0)[circle,fill]{};
		\node (n2) at (0.0,  2.0)[circle,fill]{};
		\node (n3) at (0.5,  3.0)[circle,fill]{};
		
		\node (s1) at (-1, 4)[circle,fill]{};
		\node (s2) at (0,  4)[circle,fill]{};
		\node (s3) at (1,  4)[circle,fill]{};
		
		\draw (root) -- (n1) [edgeStyle];
		\draw (root) -- (n2) [edgeStyle];
		\draw (n1) -- (n2) [edgeStyle];
		\draw (n2) -- (n3) [edgeStyle];
		\draw (n1) -- (s1) [edgeStyle];
		\draw (n3) -- (s2) [edgeStyle];
		\draw (n3) -- (s3) [edgeStyle];
		
		\draw [dashed, line width=2pt] (s2) ellipse (1.5 and 0.3);
		\node at (s3)[label={$S$},above right=5pt]{};
		\end{tikzpicture}
		\caption{The network graph, including the root and $S$.}
	\end{subfigure}
	\begin{subfigure}{0.22\textwidth}
	\centering
		\begin{tikzpicture}[scale=1.2]
		\node (root) at (0, 0)[circle,fill]{};
		\draw (root) node [right=5pt]{root};
		\node (n1) at (-0.5, 1.0)[circle,fill]{};
		\node (n2) at (0.0,  2.0)[circle,fill]{};
		\node (n3) at (0.5,  3.0)[circle,fill]{};
		
		\node (s1) at (-1, 4)[circle,fill]{};
		\node (s2) at (0,  4)[circle,fill]{};
		\node (s3) at (1,  4)[circle,fill]{};
		
		\node (t) at (0, 5)[circle, fill]{};
		\draw (t) node [right=5pt]{t};
		
		\draw (root) -- (n1) [edgeStyle] node [midway,left=1pt] {2};
		\draw (root) -- (n2) [edgeStyle] node [midway,right=1pt] {2};
		\draw (n1) -- (n2) [edgeStyle] node [midway,above=6pt] {2};
		\draw (n2) -- (n3) [edgeStyle] node [midway,right=2pt] {2};
		\draw (n1) -- (s1) [edgeStyle] node [midway,left=1pt] {2};
		\draw (n3) -- (s2) [edgeStyle] node [midway,left=1pt] {2};
		\draw (n3) -- (s3) [edgeStyle] node [midway,right=1pt] {2};
		\draw (s1) -- (t) [edgeStyle] node [midway,left=1pt] {1};
		\draw (s2) -- (t) [edgeStyle] node [midway,left=1pt] {1};
		\draw (s3) -- (t) [edgeStyle] node [midway,left=1pt] {1};
		
		\draw [dashed, line width=2pt] (s2) ellipse (1.5 and 0.3);
		\node at (s3)[label={$S$},above right=5pt]{};
		\end{tikzpicture}
		\caption{The network flow problem instance with $k=2$.}
	\end{subfigure}
	\begin{subfigure}{0.22\textwidth}
		\centering
		\begin{tikzpicture}[scale=1.2]
		\node (root) at (0, 0)[circle,fill]{};
		\draw (root) node [right=5pt]{root};
		\node (n1) at (-0.5, 1.0)[circle,fill]{};
		\node (n2) at (0.0,  2.0)[circle,fill]{};
		\node (n3) at (0.5,  3.0)[circle,fill]{};
		
		\node (s1) at (-1, 4)[circle,fill]{};
		\node (s2) at (0,  4)[circle,fill]{};
		\node (s3) at (1,  4)[circle,fill]{};
		
		\node (t) at (0, 5)[circle, fill]{};
		\draw (t) node [right=5pt]{t};
		
		\draw (root) -- (n1) [edgeStyle] node [midway,left=1pt] {2};
		\draw (root) -- (n2) [edgeStyle] node [midway,right=1pt] {1};
		\draw (n1) -- (n2) [edgeStyle] node [midway,above=6pt] {1};
		\draw (n2) -- (n3) [edgeStyle] node [midway,right=2pt] {2};
		\draw (n1) -- (s1) [edgeStyle] node [midway,left=1pt] {1};
		\draw (n3) -- (s2) [edgeStyle] node [midway,left=1pt] {1};
		\draw (n3) -- (s3) [edgeStyle] node [midway,right=1pt] {1};
		\draw (s1) -- (t) [edgeStyle] node [midway,left=1pt] {1};
		\draw (s2) -- (t) [edgeStyle] node [midway,left=1pt] {1};
		\draw (s3) -- (t) [edgeStyle] node [midway,left=1pt] {1};
		
		\draw [dashed, line width=2pt] (s2) ellipse (1.5 and 0.3);
		\node at (s3)[label={$S$},above right=5pt]{};
		\end{tikzpicture}
		\caption{A maximum flow from the root to $t$, with value $3=|S|$. All flow directions are upward.}
	\end{subfigure}
	\begin{subfigure}{0.22\textwidth}
	\centering
		\begin{tikzpicture}[scale=1.2]
		\node (root) at (0, 0)[circle,fill]{};
		\draw (root) node [right=5pt]{root};
		\node (n1) at (-0.5, 1.0)[circle,fill]{};
		\node (n2) at (0.0,  2.0)[circle,fill]{};
		\node (n3) at (0.5,  3.0)[circle,fill]{};
		
		\node (s1) at (-1, 4)[circle,fill]{};
		\node (s2) at (0,  4)[circle,fill]{};
		\node (s3) at (1,  4)[circle,fill]{};
		
		\draw [->] (-0.1, 0) -- (-0.6, 1) -- (-1, 3.9);
		\draw [->] (0.0,  0) -- (-0.4, 1) -- (-0.1, 2) -- (0.4, 3) -- (0, 3.9);
		\draw [->] (0.1,  0) -- (0.1, 2) -- (0.6, 3) -- (1, 3.9);
		
		\draw [dashed, line width=2pt] (s2) ellipse (1.5 and 0.3);
		\node at (s3)[label={$S$},above right=5pt]{};
		\end{tikzpicture}
		\caption{The paths from the root to $S$ found by the flow decomposition theorem. Each edge is used at most $k=2$ times.}
	\end{subfigure}
	\caption{An example of connection transfer paths found by our network flow approach.}
	\label{exNetworkFlow}
\end{figure}

\subsection{An Example: Full Binary Tree}

In Section \ref{eprExampleTree}, we computed the number of EPR pairs consumed when distributing a graph state among every node of a network with a full binary tree structure. If we use the same root node (the root of the tree) and paths as we introduced in Section \ref{eprExampleTree}, then the two edges connected to the root node will each be used $\frac{n-1}{2}$ times, and every other edge in the network will be used fewer times. Thus the completion time for the GST algorithm to distribute a graph state to every node of that network is $\frac{n-1}{2}$ timesteps. This is an improvement over the $n-1$ timesteps required by the EDCG algorithm.

\section{Classical Communication Requirements}

Both graphical and teleportation based graph state distribution require $O(n^2)$ bits of classical communication to distribute a graph state across a network of size $n$. The classical communication requirement comes from communicating measurement results so nodes can perform the appropriate local correction operations.


For the graphical approach---if, for each qubit whose connections we transfer to its destination node, we reorder the edge addition, $Y$-measurement, and local correction operations such that the local corrections come last\cite{measurementCalc} (see Section \ref{timeOptSec} for details on this process), then we send $O\left(n^2\right)$ classical bits for measurement results. This is because each time we transfer the connections of a qubit to its destination node, we can transmit all $O(n)$ measurement results to the root node, which will then transmit all $O(n)$ correction operation requirements (which require $O(1)$ classical communication each) to the nodes which require local correction. Thus we require $O(n)$ classical communication for each qubit whose connections we transfer to a destination node, or $O\left(n^2\right)$ communication total.

For the teleportation approach---when transferring any qubit's connections to its destination node, the local correction operation at the destination node depends on the measurement results of all of the $O(n)$ measurements done at each node along the path in the network. Hence each node that shares the final graph state requires $O(n)$ qubits of classical communication, for a total of $O(n^2)$ bits of classical communication.

The EDCG algorithm for graph state distribution also requires $O(n^2)$ classical communication. A star expansion operation on a qubit with $m$ neighbors requires $O(m)$ bits of classical communication. Distributing a GHZ state across a graph with $n$ vertices requires that each node only be communicated with once, so only $O(n)$ bits of communication are required. The EDCG algorithm requires distributing a GHZ state $n$ times, so $O(n^2)$ bits of classical communication are required. Also, performing edge measurements and subsequent local corrections to turn the EDCG state into the desired graph state requires $O(1)$ bits for each edge, for $O(n^2)$ bits total.

\begin{table*}[h]
    \begin{center}
	\begin{tabular}{|l|l|l|l|l|l|l|}
		\hline
		& EPR pairs  & Time                & \thead{Res. Graph \\ State Qubits}    & \thead{Classical \\ Comm.} & \thead{Bin. Tree \\ EPR Pairs} & \thead{Bin. Tree \\ Completion Time} \\\hline
		\thead{GST \\ algorithm} & $\frac{3n^2-2n}{8}$ & $n-1$ & $2(n-1)$                      & $O\left(n^2\right)$ & \thead{$2^{h+1}(h-1) + 2$ \\ $\hspace*{.3in} \mbox{} =\Theta(n\log n)$} & $\frac{n-1}{2}$ \\\hline
		\thead{EDCG \\ algorithm}   & $\frac{n(n-1)}{2}$  & $n-1$ & $\frac{n(n+1)}{2}$            & $O\left(n^2\right)$ & $\frac{n(n-1)}{2}$ & $n-1$ \\\hline
	\end{tabular}
	\caption{Comparison of various performance metrics between the GST algorithm and the EDCG algorithm. We also include a comparison of how both algorithms perform on the binary tree of height $h$ example given in Section \ref{eprExampleTree}.}
	\label{comparisonTable}
	\end{center}
\end{table*}

Note that both of our connection transfer approaches result in $O(1)$ bits of classical communication required for each measurement done---equivalently, $O(1)$ bits of classical communication for each EPR pair consumed. Also, the EDCG algorithm requires $\Omega(1)$ bits of classical communication per EPR pair consumed as that also requires communicating measurement results of each measurement that consumes an EPR pair. Thus Theorem \ref{eprTheorem} from Section \ref{eprSection} also extends from EPR pair consumption to bits of classical communication---if the EDCG algorithm for graph state distribution uses $c(n)$ bits of classical communication to distribute a graph state in a network of size $n$, then the GST algorithm uses $O(c(n))$ bits.

Table \ref{comparisonTable} summarizes all the cost metrics of our graph state distribution algorithm compared to the EDCG algorithm\cite{distrStates}.

\section*{Acknowledgments}
We would like to acknowledge Professor Neil Immerman for coming up with the clever network flow construction. This work was supported in part by the National Science Foundation (NSF) under Grants CNS-195744 and ERC-1941583.

\bibliography{QCE}{}
\bibliographystyle{plain}
\end{document}